\documentclass[final]{IEEEtran}

\usepackage[T1]{fontenc}
\usepackage[english]{babel}
\usepackage{amsmath,amssymb,amsfonts,amsthm}
\usepackage{bm}
\usepackage{color}
\usepackage{graphicx}
\usepackage{dsfont}
\usepackage{tabularx,booktabs}
\usepackage{url}
\usepackage{mathtools}

\usepackage[breaklinks,colorlinks,citecolor=black,linkcolor=black,urlcolor=black,
  pdftitle={},
  pdfauthor={A. M. Tillmann, C. Puchert},
  pdfsubject={},
  pdfview=FitB,
  pdfstartview=FitB]{hyperref}

\usepackage[noadjust]{cite}

\newtheorem{thm}{Theorem}
\newtheorem{rem}[thm]{Remark}

\newcommand{\Z}{\mathds{Z}}
\newcommand{\F}{\mathds{F}}
\newcommand{\ones}{\mathds{1}}

\newcommand{\st}{\text{s.t.}}
\newcommand{\norm}[1]{\lVert {#1} \rVert}

\newcommand{\abs}[1]{\lvert {#1} \rvert}

\newcommand{\define}{\coloneqq}

\newcommand{\NP}{\textsf{NP}}

\renewcommand{\P}{\textsf{P}}

\newcommand{\A}{\bm{A}}

\renewcommand{\H}{\bm{H}}

\newcommand{\G}{\bm{G}}

\newcommand{\h}{\bm{h}}

\newcommand{\w}{\bm{w}}
\newcommand{\s}{\bm{s}}
\newcommand{\x}{\bm{x}}

\newcommand{\z}{\bm{z}}
\newcommand{\0}{\bm{0}}
\newcommand{\bones}{\bm{\ones}}

\newcommand{\bgamma}{\bm{\gamma}}

\DeclareMathOperator{\rank}{rank}
\DeclareMathOperator{\supp}{supp}

\definecolor{darkgreen}{rgb}{0,0.7,0}

\date{\today}

\begin{document}

\title{Exact separation of forbidden-set cuts associated with redundant parity checks of binary linear codes}
\author{%
  Christian Puchert\IEEEmembership{} ~and~ Andreas M.
  Tillmann\IEEEmembership{}
  \thanks{Manuscript received April 7, 2020.} 
  \thanks{This work was conducted in the RWTH ERS Start-Up project
    ``Efficient exact maximum-likelihood decoding and minimum-distance
    computation for binary linear codes'' (Tillmann), funded by the
    Excellence Initiative of the German federal and state governments.}%
  \thanks{C. Puchert was with the Chair of Operations Research at RWTH
    Aachen University, Kackertstr. 7, 52072 Aachen, Germany (e-mail:
    puchert@or.rwth-aachen.de), and is now with Ab Ovo Business and
    Software Solutions, D\"usseldorf, Germany.}%
  \thanks{A. M. Tillmann was with the Chair of Operations Research and the
    Visual Computing Institute at RWTH Aachen University, Germany, and is
    now with the Institute for Mathematical Optimization, Technische
    Universit\"at Braunschweig, Universit\"atsplatz 2, 38106 Braunschweig,
    Germany (e-mail: a.tillmann@tu-bs.de).}
    \thanks{This work has been
      submitted to the IEEE for possible publication. Copyright may be
      transferred without notice, after which this version may no longer be
      accessible.}%
}

\markboth
{Exact separation of forbidden-set cuts associated with redundant parity checks of binary linear codes}
{Exact separation of forbidden-set cuts associated with redundant parity checks of binary linear codes}

\maketitle

\begin{abstract}
  In recent years, several integer programming (IP) approaches were
  developed for maximum-likelihood decoding and minimum distance
  computation for binary linear codes. Two aspects in particular have been
  demonstrated to improve the performance of IP solvers as well as adaptive
  linear programming decoders: the dynamic generation of forbidden-set
  (FS) inequalities, a family of valid cutting planes, and the utilization
  of so-called redundant parity-checks (RPCs). However, to date, it had
  remained unclear how to solve the exact RPC separation problem (i.e., to
  determine whether or not there exists any violated FS inequality
  w.r.t. any known or unknown parity-check). In this note, we prove
  \NP-hardness of this problem. Moreover, we formulate an IP model that
  combines the search for most violated FS cuts with the generation of
  RPCs, and report on computational experiments. Empirically, for various
  instances of the minimum distance problem, it turns out that while
  utilizing the exact separation~IP does not appear to provide a
  computational advantage, it can apparently be avoided altogether by
  combining heuristics to generate RPC-based cuts.
\end{abstract}



\IEEEpeerreviewmaketitle

\section{Introduction \& Preliminaries}\label{sec:intro}
\IEEEPARstart{T}{he} well-known minimum distance computation problem for
binary linear codes, defined by a parity-check matrix
$\H\in\{0,1\}^{m\times n}$, can be stated as
\begin{align}
  \label{eq:mindist} \min\quad &\norm{\x}_0\\
  \nonumber \st\quad &\H\x=\0\mod 2,~\x\neq \0,~\x\in\{0,1\}^n,
\end{align}
where $\norm{\cdot}_0$ denotes the Hamming weight (i.e., number of
nonzeros). Similarly, the maximum-likelihood decoding (MLD)
problem is obtained by replacing the objective function with
$\bgamma^\top\x$, where $\bgamma$ is a (known) vector of negative
log-likelihoods (see, e.g., \cite{DraperYW2007}). Both
problem~\eqref{eq:mindist} and MLD are \NP-hard, see~\cite{Vardy1997}
and~\cite{BerlekampMcEvT1978}, respectively. For simplicity, we will
focus on~\eqref{eq:mindist} throughout, but emphasize that as we will
mainly deal with issues pertaining to the feasible set, the results are
also applicable in the MLD context.

Problem~\eqref{eq:mindist} (and analogously, the MLD problem) can be
rewritten in several ways (see, e.g.,
\cite{KehaD2010,PunekarKWTRH2010,TanatmisRHPKW2009,YangWF2007,KabakulakTP2018});
a fairly straightforward, and the most common generic approach
(cf.~\cite{BreitbachBLK1998,TanatmisRHPKW2010,SchollKHR2013}) is to
reformulate~\eqref{eq:mindist} as the integer program
\begin{align}
  \label{eq:mindistbasicip}
  \min\quad &\bones^\top \x\\
  \nonumber \st\quad &\H\x -2\z = \0,~\bones^\top \x\geq
                       1,~\x\in\{0,1\}^n,~\z\in\Z^n_{\geq 0},
\end{align}
where $\bones$ is the all-ones vector. Indeed, it is easily seen that the
auxiliary nonnegative integer variables $\z$ enforce the required (even)
codeword parity in all rows of the equality constraints.

Modern IP solvers are based on effective combination of the
branch-and-bound paradigm (creating subproblems by fixing some variables
and pruning parts of the resulting search tree, e.g., if local objective
bounds prevent further improvements in the respective subtree) and the use
of (linear) cutting planes, i.e., inequalities that are valid for integral
feasible points but may be violated by fractional solutions of the commonly
employed linear programming (LP) relaxations.  Polyhedral investigations of
the codeword polytope (i.e., the convex hull of integer feasible points
of~\eqref{eq:mindist}) in
\cite{BarahonaG1986,GroetschelT1989b,Tanatmis2011,Jeroslow1975,FeldmanWK2005,
  YangWF2007,HelmlingRT2012,LanciaS2018} (and other works) have identified,
in particular, the following class of valid inequalities to strengthen the
LP relaxation and to improve solver performance when tackling problems
like~\eqref{eq:mindistbasicip}:
\begin{equation}\label{eq:FScut}
  \sum_{j \in S} x_{j} - \sum_{j \in \supp(\h) \setminus S} x_{j} \leq |S| - 1,
\end{equation}
where
\mbox{$\emptyset\neq S\subset\supp(\h)\define\{j\in\{1,2,\dots,n\}:h_j\neq 0\}$}
is odd and $\h$ is a parity-check. These inequalities became known as (odd)
forbidden-set (FS) inequalities. It has been established that for any given
parity-check, at most one of the associated FS inequalities can be violated
at a time (by a given point)~\cite{TaghaviS2008}, and that if a violated
one (called a \emph{cut}) exists, it can be found in polynomial time; the
state-of-the-art method for the latter is the sorting-based routine
described in~\cite{ZhangS2012} (see also~\cite{FalsafeinM2016}).

Moreover, several papers discussing IP or adaptive LP approaches
for~\eqref{eq:mindist} or MLD incorporate heuristics to find new redundant
parity-checks (RPCs)---i.e., some linear combinations (over $\F_2$, i.e.,
in modulo-2 arithmetic) of rows of~$\H$---that hopefully yield violated
FS inequalities, see, for instance,
\cite{TaghaviS2008,MiwaWT2009,ZhangS2012,FalsafeinM2016}.
While some results exist on when such heuristics will succeed in finding a
cut that is violated by a given (fractional) point (cf., e.g.,
\cite{TaghaviS2008,TanatmisRHPKW2010}), 
there is no general guarantee for this. Indeed, it has often been claimed
that the problem of finding cut-inducing RPCs is intractable (see, e.g.,
\cite{TanatmisRHPKW2010,FalsafeinM2016,TaghaviS2008}), but a formal proof
of this statement appears to be missing so far.

In the following, we close this gap by providing an \NP-hardness proof of
the RPC separation problem. Moreover, we will introduce an IP model to
solve it, i.e., to find a maximally violated FS cut among all possible FS
inequalities w.r.t.\ all parity-checks, known or unknown (also allowing to
conclude that all FS inequalities are satisfied, should this be the case),
and assess its potential practical applicability in some computational
experiments with problem~\eqref{eq:mindistbasicip}.

\section{Exact RPC Cut Generation/Separation}
\noindent
Clearly, if for a given $\x^*\in[0,1]^n$, not all FS
inequalities~\eqref{eq:FScut} are satisified, there exists at least one
parity-check~$\h$ and odd index set $S\subset\supp(\h)$ such that the
constraint violation
\[
-|S|+1 +\sum_{j\in S}x^*_{j} - \sum_{j\in\supp(\h)\setminus S}x^*_{j} > 0.
\]
Thus, the \emph{RPC separation problem} can be stated as follows:
\emph{Given $\x^*\in[0,1]^n$, find a valid parity-check that yields a
  maximally violated FS cut, or conclude that all FS inequalities are
  satisfied.}  This task can be formulated as the optimization problem
\begin{align}\label{RPCsepaprob}
  \max\quad &-\abs{S}+1+\sum_{j\in S}x^*_j -\sum_{j\in\supp(\h)\setminus S}x^*_j\\
\nonumber  \text{s.t.}\quad &\h\text{ is a parity-check,}\\
\nonumber  &\emptyset\neq S\subseteq\supp(\h),~\abs{S}\text{ odd};
\end{align}
clearly, the optimal objective value is strictly larger than zero
if and only if a cut-inducing parity-check $\h$ exists. The objective
is the violation of the respective FS inequality associated with~$\h$
and an odd subset $S$ of its support at the given point~$\x^*$, and the
optimization problem searches for the largest such violation among all
valid parity-checks.

With a few observations, we can turn the above abstract optimization
problem into an equivalent concrete IP: First, recall that parity-checks
correspond to dual codewords, cf.~\cite{TaghaviS2008}; thus, we can express
the first constraint in terms of a generator matrix $\G$ of the code (which
can be obtained as a basis matrix for the $\F_2$-nullspace of the given
parity-check matrix $\H$ by Gaussian elimination) as $\G^\top \h=\0\mod 2$
with $\h\in\{0,1\}^n$. To model the odd subsets $S$ of $\supp(\h)$, we
introduce binary variables $\s\leq \h$ with the constraint that
$\ones^\top \s$ is odd, i.e., $\ones^\top \s=1\mod 2$. The condition
$S\neq\emptyset$ is then reflected as $\ones^\top s\geq 1$, and is
automatically ensured by the previous odd-sum constraint. The objective
of~\eqref{RPCsepaprob} can be rewritten as
\begin{align*}
    &-\abs{S}+1-\sum_{j\in\supp(\h)}x^*_j+2\sum_{j\in S}x^*_j\\
  = &-\ones^\top \s +1 -(\x^*)^\top \h+2(\x^*)^\top(\s) \\
  = &(2\x^*-\ones)^\top \s-(\x^*)^\top \h+1.    
\end{align*}
Assuming w.l.o.g.\ that $\rank(\H)=m<n$ (so $\G\in\{0,1\}^{n\times(n-m)}$
with rank $n-m$), with auxiliary integer variables
$\z^h\in\Z^{n-m}_{\geq 0}$ and $z^s\in\Z_{\geq 0}$ to linearize the modulo-2
equations, we thus arrive at the \emph{RPC separation IP} (equivalent
to~\eqref{RPCsepaprob})
\begin{align}\label{RPCsepaIP}
  \max\quad &(2\x^*-\ones)^\top \s-(\x^*)^\top \h+1\\
\nonumber  \text{s.t.}\quad &\G^\top \h -2\z^h =\0\\
\nonumber            &\ones^\top \s -2z^s=1\\
  \nonumber            &\s\leq \h\\
  \nonumber            &\h,\s\in\{0,1\}^n,~\z^h\in\Z^{n-m}_{\geq 0},~z^s\in\Z_{\geq 0}.
\end{align}
(An alternative formulation that does not require knowledge, or
construction, of the generator matrix $\G$ can be obtained by replacing the
first equality constraint by $\h=\H^\top \w - 2\z^h$ with $\w\in\{0,1\}^m$
and $\z^h\in\Z^{n}_{\geq 0}$; however, this variant has $m$ additional
binary variables and $m$ more integer variables.)

Our following result establishes that, in general, exact RPC
separation is computationally challenging:
\begin{thm}\label{thm:RPCsepaNPhard}
  The RPC separation problem is \NP-hard.
\end{thm}
\begin{proof}
  The \NP-hardness is due to the fact that the separation problem contains
  as a special case the task of finding the girth of a binary matroid
  (i.e., the minimum distance problem~\eqref{eq:mindist}), which is well-known to be
  \NP-hard~\cite{Vardy1997}: 
  Let $\A\in\{0,1\}^{m\times n}$ be an arbitrary instance of the binary
  matroid girth problem; w.l.o.g., we may assume $\rank(\A)=m<n$ (over
  $\F_2$). With $\x^*\define(1/2)\ones$ and $\G^\top\define\A$,
  the corresponding instance of the RPC separation problem (in IP
  form~\eqref{RPCsepaIP}) collapses to
  \begin{align*}
    \max_{\s,\h\in\{0,1\}^n} &-\tfrac{1}{2}\ones^\top \h +1\\
    \st\quad &\A \h=\0\mod 2,~\s\leq \h,~\ones^\top \s=1\mod 2\\
    \Leftrightarrow\quad 1-\tfrac{1}{2}&\min_{\h\in\{0,1\}^n}\left\{\ones^\top \h\,:\,\A \h=\0\mod 2,~\h\neq \0\right\}.
  \end{align*}
  The latter minimization problem is precisely the girth problem for the
  instance given by $\A$, so unless $\P=\NP$, there cannot exist a
  polynomial-time algorithm for the RPC separation problem.
\end{proof}
\begin{rem}
  Note that, although the above proof shows hardness of a separation
  problem for a large class of valid inequalities for the minimum distance
  problem~\eqref{eq:mindist} (and others) by exhibiting that very same
  problem as a special case, the result is \emph{not} implied by the
  general so-called ``equivalence of optimization and separation''
  (cf.~\cite{GroetschelLS1993}). Indeed, even if \emph{all} FS inequalities
  for \emph{all} parity-checks (given and redundant ones) were added to the
  IP model~\eqref{eq:mindistbasicip}, its LP relaxation is not guaranteed
  to yield an integral (binary) optimum point~\cite{ZumbraegelSF2012}.

  Moreover, using results from~\cite{McGregorM2010}, \NP-hardness of the
  RPC separation problem can be seen to persist even if $\A$ is restricted
  to be the parity-check matrix of an LDPC code.
\end{rem}

\begin{rem}
  The problem obtained by minimizing the RPC separation IP's objective
  (instead of maximizing it) is, in fact, also \NP-hard, as it contains the
  well-known max-cut problem as a special case (similar to the problem of
  finding a cycle of maximum weight in a binary matroid,
  cf.~\cite{GroetschelT1989b}); we omit the proof for brevity. The fact
  that both maximizing and minimizing the objective function over the
  feasible set of the RPC separation problem are \NP-hard indicates that
  the hardness is inherent in the combinatorial constraints. (Similarly,
  the paper \cite{NtafosH1981} gave a hardness-of-maximization result
  complementing the intractability of the minimization
  problem~\eqref{eq:mindist} shown
  later 
  in~\cite{Vardy1997}.)
 \end{rem}

 Naturally, given its inherent intractability as established by
 Theorem~\ref{thm:RPCsepaNPhard}, solving the RPC separation
 IP~\eqref{RPCsepaIP} will be more challenging than running the heuristic
 schemes from, e.g., \cite{ZhangS2012,FalsafeinM2016}. Ideally, the
 additional computational effort to gain stronger cuts (and corresponding
 ``optimal'' new parity-checks) could lead to a sufficient benefit in terms
 of search space reduction (particularly, LP relaxation tightening) to
 compensate for the increased runtime overhead that can be expected. It is
 also worth pointing out that the IP~\eqref{RPCsepaIP} may also be used in
 a heuristic fashion: any feasible solution with positive objective value
 yields a violated FS inequality, so the solving process may be terminated
 early as soon as the current objective value turns positive, still
 providing a cut (although not necessarily a most violated one).

\section{Computational Experiments}\label{sec:experiments}
\noindent
We now turn to some numerical experiments, based on the minimum distance
IP~\eqref{eq:mindistbasicip}, to assess the influence of separating FS
inequalities and redundant parity-checks in different ways. We employ
the open-source MIP solver SCIP~\cite{scip} (with the LP solver SoPlex that
comes with it) and use a variety of test instances from the Channel Code
Database~\cite{channelcodes}. Our test set consists of 27 parity-check
matrices, of which 5 belong to array codes, 3 to BCH/Hamming codes, 6 to
LDPC codes, 2 to polar codes, one is the Tanner(3,5) code, 5 belong to
WiMax and 3 to WiMax-like codes, and 2 to WRAN codes. The average matrix
size is about $106\times 372$, ranging from $8$ to $588$ rows
(parity-checks) and from $32$ to $1344$ variables.
For brevity, we forego instance-specific details.

All experiments were carried out in single-thread mode on a Linux machine
with Intel Core i7-7700T CPUs (2.9\,GHz, 8\,MB cache) and 16\,GB memory,
using SCIP~6.0.2. We set a time limit of 1 hour (3600\,s) per instance.

The state-of-the-art separation routine from~\cite{ZhangS2012} for FS
inequalities associated with \emph{given} parity checks is, in fact,
implemented in SCIP already (as part of the XOR constraint handler). We
extended the implementation by the best-known RPC cut generation heuristics
from~\cite{ZhangS2012} and from~\cite{FalsafeinM2016}, respectively, as
well as by an exact separation routine based on the RPC separation
IP~\eqref{RPCsepaIP}. 

First, we used SCIP as a black-box, applied to the
model~\eqref{eq:mindistbasicip}; by default, SCIP does not separate FS
inequalities. Then, with separation activated in the root node only, in
every node, or on every 5-th level of the branch and bound tree,
respectively, we tested the following solver variations:
\begin{enumerate}
  \item exact separation of FS inequalities for
    existing parity-checks according to~\cite{ZhangS2012} (``ZS'', for short),
  \item ZS, followed by exact RPC cut separation via \eqref{RPCsepaIP}, if
    ZS found no violated FS inequality (``ZSe''),
  \item ZS, followed (if unsuccessful) by the RPC cut heuristic from
    \cite{ZhangS2012}, followed (if also unsuccessful) by the exact
    separation routine (``ZS+e''),
  \item ZS, followed (if unsuccessful) by the RPC cut heuristic from
    \cite{ZhangS2012}, then (if still unsuccessful) by that from
    \cite{FalsafeinM2016}, and finally (if still unsuccessful) by
    the exact separation routine (``ZS++e'').
\end{enumerate}
We 
leave all of SCIP's many other solver settings (w.r.t. general-purpose
heuristics, cutting planes, branching rules etc.) at their respective
default values. The solver also automatically takes care of cutting plane
management, i.e., we neither force inequalities into the LP nor do we
remove them by hand, but let the solver decide how to handle and store
cuts. (Nevertheless, FS inequalities are given the highest separation
priority, so that the solver first tries to find cuts of this type.)

Table~\ref{tab:results1} presents average results over all instances,
namely the arithmetic mean optimality gap ((upper bound $-$ lower
bound)/(lower bound)$\cdot 100$\%), the number of instances (out of 27)
that were solved to optimality within the given time limit, and the average
number of processed branch-and-bound (search tree) nodes as well as the
average running times (in seconds) of the different solver variants. For
nodes and times, we state shifted geometric mean values with shifts 100
and~10, respectively, to mitigate the influence of easy instances.

\begin{table}
  \caption{Overview of experimental results.}
  \label{tab:results1}
  \begin{center}
    \begin{tabular*}{\columnwidth}{@{\extracolsep{\fill}}@{\quad}l@{\quad}r@{\quad}r@{\quad}r@{\qquad}r@{\quad}}
      \toprule
      solver variant          & gap (\%) & \# solved & nodes & time (s) \\
      \midrule
      SCIP default            & 15.92         & 21        & 74464.3    & 112.9 \\[0.5em]
      ZS (sepa. root only)    & 12.37         & 21        & 81829.1    & 116.9 \\
      ZS (sepa. always)       &  1.45         & 26        &  4394.1    &  84.1 \\
      ZS (sepa. freq. 5)      &  1.28         & 26        & 14518.6    &  70.9 \\[0.5em]
      ZSe (sepa. root only)   & 12.37         & 21        & 81833.0    & 117.1 \\
      ZSe (sepa. always)      & 1526.64       &  6        &   145.2    & 1562.4 \\
      ZSe (sepa. freq. 5)     & 557.50        & 10        &  836.8     & 837.7  \\[0.5em]
      ZS+e (sepa. root only)  & 12.37         & 21        & 81816.9    & 117.1 \\
      ZS+e (sepa. always)     & 12.63         & 19        &  1517.0    & 468.6 \\
      ZS+e (sepa. freq. 5)    & 12.25         & 21        &  7751.2    & 293.4 \\[0.5em]
      ZS++e (sepa. root only) & 17.60         & 21        & 150529.4   & 182.8 \\
      ZS++e (sepa. always)    &  5.43         & 24        &  2564.8    & 128.5 \\
      ZS++e (sepa. freq. 5)   &  7.76         & 23        &  9256.5    & 108.9 \\
      \bottomrule            
    \end{tabular*}
  \end{center}
\end{table}

First of all, the overview in Table~\ref{tab:results1} clearly confirms
previous studies: adding FS cuts during the solving process can
signi\-ficantly help solve IPs like~\eqref{eq:mindistbasicip}. While adding
cuts only in the root node can be seen to (adaptively) strengthen the
original LP relaxation, the effect becomes much more prominent when
parity-checks are also used to separate cuts in deeper levels of the search
tree. Comparing the ``SCIP default'' results with, in particular, the
``ZS'' results, using (exact) separation of FS cuts based on the
\emph{known}, original parity-checks allows for more instances to be solved
(thus also yielding a smaller average optimality gap) in significantly
shorter time.

The results also show that more is not necessarily better: adding cuts at
every search node does lead to the smallest average number of nodes that
are explored before certifying optimality (or reaching the time limit), but
apparently also significantly increases the time needed to solve the many
(LP) subproblems during the branch-and-bound process. Separating FS cuts
only at every 5-th level of the search tree typically yields results of
comparable quality in a shorter average time, despite requiring more nodes
to be explored.

The results incorporating the generation of redundant parity-checks in
order to find further violated FS inequalities are somewhat
ambiguous. Recall that we only ever used RPC-based schemes if the basic FS
cut separation failed. Thus, the longer runtimes and less impressive gains
in solution quality of the variants ``ZSe'', ``ZS+e'' and ``ZS++e''
compared to ``ZS'' indicate that, on average, it does not pay off to search
for FS cuts by exploring RPCs if the originally given parity-checks do not
already yield such cuts. Nevertheless, it is important to emphasize that
this is an empirical average statement only -- indeed, e.g., for the
Tanner(3,5) instance, the fastest ``ZS'' version (in this case, the one
separating FS cuts in all nodes) solved the minimum distance problem in
about 2029\,s, whereas the ``ZS++e'' variant (also with separation in every
node) only took roughly 1417\,s. Thus, our results do not directly
contradict earlier claims that RPC cuts help, although their benefit
does appear to be less pronounced when integrating the corresponding separation
schemes into a full MIP framework rather than more basic branch-and-cut or
adaptive LP decoding procedures that had been explored~be\-fore. (Recall
also that \cite{TaghaviS2008} already noted that the relative improvement
provided by RPC cuts decreases as the number of variables increases, which
might have played a role here as well.)

Finally, the results for ``ZSe'' and ``ZS+e'' show that, unfortunately
(though not unexpectedly, given Theorem~\ref{thm:RPCsepaNPhard}), exact RPC
cut separation does not pay off, as the IP separation subproblems are too
hard in practice to be routinely solved within a branch-and-cut framework
for the original problem. On the positive side, we note that in the variant
``ZS++e'', i.e., where \emph{exact} RPC separation would only be used if both
state-of-the-art RPC cut \emph{heuristics} failed, the separation
IP~\eqref{RPCsepaIP} was actually \emph{never} called at all. Thus,
although one cannot be sure that the most violated cut was used, some
violated cut was always found without having to resort to exact separation.
Overall, however, the best results were still achieved by variant ``ZS''
that did not use RPC cuts at all.

\section{Concluding Remarks}
\noindent
In this note, we formally proved computational intractability of the RPC
cut separation problem for binary linear codes. Moreover, our numerical
experiments demonstrated that an associated IP formulation is indeed too
hard to solve in practice to become useful during branch-and-cut
methods. While our empirical results further showed that by combining the
existing state-of-the-art RPC separation heuristics, the exact problem
never needed to be resorted to, they also suggested that, on average, RPC
cuts do not help improve the solution process of a current powerful MIP
solver compared to using only FS cuts associated with the initial, known
parity-checks.

\bibliographystyle{IEEEtran}
\bibliography{puchert_tillmann_2020}

\end{document}